\documentclass[12pt]{article}

\usepackage{times}

\usepackage{placeins}

\usepackage{mathtools}

\usepackage{cite}

\usepackage{lineno}

\newenvironment{sciabstract}{%
\begin{quote} \bf}
{\end{quote}}

\title{A deterministic balancing score algorithm to avoid common pitfalls of propensity score matching} 

\author
{Felix Bestehorn$^{1\ast}$, Maike Bestehorn$^{2}$, Markus Bestehorn$^{3}$, Christian Kirches$^{1}$\\
\\
\normalsize{$^{1}$Institute for Mathematical Optimization, Technische Universit\"at Braunschweig}\\
\normalsize{Universit\"atsplatz 2, 38106, Braunschweig, Germany}\\
\normalsize{$^2$ProMedCon GmbH, Lechnerstrasse 19, 82067, Sch\"aftlarn, Germany}\\
\normalsize{$^3$Riedpark 25, 6300 Zug, Switzerland}
\\
\normalsize{$^\ast$To whom correspondence should be addressed; E-mail:  f.bestehorn@tu-bs.de.}
}

\date{}
\usepackage{moreverb}
\usepackage{hyperref}

\usepackage{lineno}

\usepackage{lipsum}
\usepackage{amssymb}
\usepackage{amsmath}
\usepackage{algorithm}
\usepackage{algorithmicx}
\usepackage{algpseudocode}
\hypersetup{
	pdfborder={0 0 0},
	colorlinks = false,
	hidelinks=true
	pdflinkmargin=0pt,
}

\newtheorem{prop}{Proposition}
\newtheorem{theorem}{Theorem}

\newtheorem{property}{Property}
\newtheorem{definition}{Definition}
\newtheorem{example}{Example}

\newenvironment{proof}
{\textit{Proof:} }
{\hfill $\square$}
\begin{document} 

\baselineskip16pt

% Make the title.

\maketitle

\begin{sciabstract}
Propensity score matching (PSM) is the de-facto standard for estimating causal effects in observational studies. We show that PSM and its implementations are susceptible to several major drawbacks and illustrate these findings using a case study with $17,427$ patients. We derive four formal properties an optimal statistical matching algorithm should meet, and propose Deterministic Balancing Score exact Matching (DBSeM) which meets the aforementioned properties for an exact matching. Furthermore, we investigate one of the main problems of PSM, that is that common PSM results in one valid set of matched pairs or a bootstrapped PSM in a selection of possible valid sets of matched pairs. For exact matchings we provide the mathematical proof, that DBSeM, as a result, delivers the expected value of all valid sets of matched pairs for the investigated dataset.
\end{sciabstract}

%\linenumbers

\section{Introduction}\label{sec:intro}

Statistical matching (SM) is widely used to reduce the effect of confounding~\cite{Rubin1973,Anderson1980,Kupper1981} when evaluating the relative effects of two different paths of action in an observational study. For instance, medical studies use SM to compare mortality rates between two patient populations that have received two different treatments or procedures~\cite{Ray2012,Zhang2015,Gozalo2015,Zhang2016,Cho2016,Bruno2017,Nichay2017,Burden2017,Mcevoy2016,Schermerhorn2008,Lee2017,Capucci2017,Tranchart2016,Zangbar2016,Dou2017,Fukami2017,McDonald2017,Lai2016,Abidov2005,Adams2017,Kishimoto2017,Kong2017,Chen2016,Seung2008,Shaw2008,Liu2016,Svanstrom2013,Salati2017}. With more than $16{,}000$ citations in research papers within the last 12 months \cite{King2016}, Propensity Score Matching (PSM)~\cite{Rosenbaum1983} is the de-facto standard for SM in such applications.

While some limitations of PSM have been studied~\cite{King2016, Austin2011, Pearl2009} and the quality of PSM results have been discussed through empirical evaluations, PSM results have not yet been sufficiently investigated in a mathematical sense. This is particularly important since PSM results are often used for making critical decisions such as choosing the best medical procedure. On the basis of a general PSM algorithm we show, that PSM can lead to arbitrary decision making and that PSM-based results are susceptible to manipulation by cherry-picking outcomes supporting certain hypotheses. We illustrate our findings using the example of a real-world medical study and present  Deterministic Balancing Score Exact Matching (DBSeM) - a new approach for exact SM that delivers the average result for all valid sets of exact matchings for the investigated dataset and is therefore reproducible and reliable.

Specifically, we make the following contributions:

$C1$\label{C1}: We investigate potential pitfalls based on an analysis of general PSM implementations taken from guidelines for implementing PSM algorithm and illustrate our findings by using the database for isolated aortic valve procedures $2013$ containing information on $17{,}427$ patients, their treatment and various other, relevant parameters. %For instance, this shows that the ordering of the input data may significantly change the result obtained through PSM.

$C2$\label{C2}: We formally derive four properties that an optimal SM algorithm has to meet: reproducibility of results, order-independence, data-completeness, and conservation. PSM does not have these properties. 

$C3$\label{C3}: We introduce DBSeM as a clustering-based SM approach and prove that DBSeM satisfies the four properties of an optimal SM algorithm.

$C4$\label{C4}: We show that bootstrapped PSM results converge towards the results gained by DBSeM, which is the average result of all sets of exact matched pairs.

The motivation behind our contributions are to develop an algorithm, which is usable for statistical matching in general and is deterministic. Fulfilling the deterministic property is important for the algorithm as results obtained through application of a deterministic algorithm can be reproduced by fellow researchers, leading to verifiability of results as well as to further common ground for scientific discussion in the field of observational studies. 

Note that this is a mathematical article. Hence the proven results are generally applicable to all datasets used in exact SM.

While this paper uses medical terminology such as "patients" or "treatment" to illustrate its content, our results are applicable to other fields of research with observational studies as well.

\section{Related work and definitions}\label{sec:related}

In the context of medical observational studies the propensity score (PS) is the probability that a patient is assigned to a particular treatment given a vector of observed covariates~\cite{Rosenbaum1983}. PSM matches patients with similar/equal PS to allow a comparison between treatment results. Thus PS and PSM are defined by Rosenbaum and Rubin in~\cite{Rosenbaum1983} as follows:

\noindent Given a set $G := \{x_1,\,\ldots,\,x_a,\,z_1,\,\ldots,\,z_b\}$ of patients. Let $A := \{x_1,\,\ldots,\,x_a\}$ and $B := \{z_1,\,\ldots,\,z_b\}$ be the patient partition for the respective treatments. The $s$ statistically relevant properties -- covariates -- of each patient $p \in G$ are specified by an $s$-dimensional covariate vector $cv(p)\in \mathbb{R}_{\geq 0}^s$ and the observed result is identified by $obs(p) \in \mathbb{R}$. 

In randomized studies the PS is known by design, whereas in non-randomized studies -- the case of the illustrative example -- it needs to be estimated from the dataset. PS is typically~\cite{Stuart2014} estimated using logistic regression, but can also be calculated through other regression methods such as probit, tobit or cox regression, with treatment as the dependent variable and covariates as baseline. Given regression coefficients $\beta_j,\, 0\leq j\leq s$, from the logistic regression the estimated PS of a patient $p$ is defined as 
\begin{equation}
\label{eq:prop_score}
ps(p) := \frac{e^{\beta_0 + \sum_{j = 1}^{s}\beta_{j}cv_j(p)}}{1+e^{\beta_0 + \sum_{j = 1}^{s}\beta_{j}cv_{j}(p)}}.
\end{equation} 

To compare patients with each other one can now compute the estimated propensity~score~differences~(PSD) from the PS of all patients for the dataset as 

\begin{equation}
\label{eq:psd}
psd_{i,\,j} := \vert ps(x_i) - ps(z_j)\vert,\, \forall 1\leq i \leq a,\,1\leq j \leq b. 
\end{equation}

Finally one has to match the patients and in general there are two classes of SM, 
that are used to match members of different sets, i.e., patients:
\begin{itemize}
	\item Exact matching \cite{Stuart2010,Iacus2011}: Only members of different sets with equal covariate vectors are matched, i.e., for PSM $psd_{i,\,j} = 0$.
	\item $\delta$-matching \cite{Stuart2010}: Members of different sets can be matched if they are similar enough according to a chosen similarity measure, e.g., Mahalanobis distance \cite{Iacus2011} or for $\delta$-PSM $psd_{i,\,j} \leq \delta$.
\end{itemize}

Different algorithmic realizations of $\delta$-matching are for example caliper matching, nearest neighbor matching or optimal matching~\cite{Stuart2010,Caliendo2005,Rosenbaum1989}.
\newline

\noindent The foundation for both, exact and $\delta$-PSM was laid by Rubin and Rosenbaum~\cite{Rosenbaum1983}, by introducing the notion of balancing scores. A balancing score $b(cv(p))$ of a patient is a value assignment, such that the conditional distribution of $cv(p)$ is the same for patients $p$ from both treatment groups, $A$ and $B$. Rubin and Rosenbaum showed that PS is the coarsest balancing score, while $cv(p)$ is the finest~(\cite{Rosenbaum1983}, section $2$) and that if treatment assignment is strongly ignorable, then the difference between the two respective treatments is an unbiased estimate of the average treatment effect at that balancing score value~(\cite{Rosenbaum1983}, theorem $3$).

We will, if not stated otherwise, only consider exact PSM in this paper. Besides ease of presentation our reasons are manifold:

\begin{enumerate}
	\item An ideal experimental design would be to compare the outcome of two therapies for pairs of patients with exactly the same condition vector. For this reason we focus on exact matching in this paper. Additionally exact matching is the best possible type of of $\delta$-PSM~\cite{Rubin1974}.
	\item \label{item:ex_1} Exact matching is a special case of the more general $\delta$-matching. Thus every $\delta$-matching contains an exact matching or at least the attempt of an exact matching on a subset of patients and pitfalls emerging in exact matching are present in $\delta$-matching as well.
	\item \label{item:ex_2} If pitfalls are present in exact matching, then letting $\delta >0$, either amplifies the effects of these pitfalls or does not affect them in any way. Most importantly the pitfalls do not vanish.
	\item \label{item:ex_3} Pitfalls emerging in exact matching are significant for the whole theory of PSM, as the best case for SM is a dataset, which is fully matchable by exact matching.
	\item If no exact matches between two therapy groups exist, then the question of comparability of the two groups on the basis of the given dataset arises as they have no common support.
\end{enumerate}

Note that, because of reasons~\ref{item:ex_1}--\ref{item:ex_3}, considering only exact matching does not impair the scope of our deductions regarding the implications for $\delta$-matching.

Additionally we limit the presentation to $1$:$1$ exact matchings as $1$:$1$ matching procedures have the highest amount of possible matchings for fixed match-sizes and all possible $k$:$l$ matchings are included in the set of possible $1$:$1$ matchings, see subsection~\ref{subsec:many-one} for further explanations regarding $k$:$l$ and one-to-many PSM.

Algorithm~\ref{alg:PSM} describes the general structure of an $1$:$1$ PSM-based matching procedure (cf.~\cite{Caliendo2005}): \FloatBarrier
\setlength{\intextsep}{7.5pt}
\begin{algorithm}
\caption{General $1$:$1$ PSM-based/statistical matching procedure}
\label{alg:PSM}
\begin{algorithmic}[1]
\State  \label{state:psm} Compute $psd_{i,\,j}\,\forall 1\leq i \leq a,\,1\leq j \leq b$ (e.g., using logistic or tobit regression).
\State \label{state:balance} Check balancing of propensity score (e.g., known covariates of high influence should have high influence on the regression value).
\For{each patient $x_i \in A\, (1\leq i \leq a)$} \label{state:PSM_match}
	\Statex Create Matching Set $M_i = \emptyset$.
	\Statex Search for unmatched patient $z_j \in B\, (1\leq j \leq b)$ with $psd(i,\,j) \equiv 0$.
	\Statex \textbf{If} $z_j \in B$ was found in previous step: Set $M_i := \{x_i,\,z_j\}$
	\Statex	Continue with next patient from $A$.
\EndFor
\State \label{state:covariate} Check covariate balancing in matches and matching quality (e.g., homogenization) and output matching sets $M_i$.	
\end{algorithmic}
\end{algorithm}
\FloatBarrier

Steps \ref{state:balance} and \ref{state:covariate} do not have to be considered in this paper because exact matching -- if viable -- completely balances covariates and achieves complete harmonization.

%TODO fix the citations here
Furthermore the various matching strategies applicable in step $3$, such as nearest neighbor~\cite{Stuart2010}, stratification~\cite{Iacus2011} or optimal~\cite{Rosenbaum1989} matching, are irrelevant for this paper. This is because each strategy's strengths and weaknesses come to fruition in exact matching as $psd_{i,\,j}=0$ (and $cv(x_i) \equiv cv(y_j)$) is either true for all PSM strategies or for none. 

\subsection{$1$:$2$ and one-to-many PSM}\label{subsec:many-one}
$1$:$2$~PSM is a variant of PSM were one patient from one therapy group gets matched to two patients from the other therapy group, if there are two patients meeting the matching criteria. This leads to a loss of information as possible matchings could be ignored. For instance, let $x_i$ be an arbitrary patient of $A$ and there exist no other patients in $A$ with the same PS. Assuming that there are ten patients in $B$ with the same PS as $x_i$, there are
$(\begin{smallmatrix}
10\\
2
\end{smallmatrix}) = 45$ many possible $1$:$2$ matchings out of which only one gets chosen, while the information in the remaining eight unmatched patients gets lost. Note that this can happen in $\delta$-PSM as well. 

Obviously this behavior persists in the general case of $k$:$l$ PSM, where $k,\,l\in \mathbb{N}$, $k$ patients from one therapy group get matched to $l$ patients of the other group, if all patients meet the matching criteria. Consequently we will not consider one-to-many or its more general case of $k$:$l$ PSM in this article, see also subsection \ref{subsec:incomplete} on incomplete usage of data.

\subsection{Bootstrapping}\label{subsec:bootstrap}
Bootstrapping techniques \cite{Austin2014} are applied in PSM to avoid negative effects occurring due to randomness or statistical outliers. Considering the example from the previous subsection \ref{subsec:many-one} again:
Let $x_i$ be an arbitrary element of $A$ and there exist ten patients from $B$ with equal PS. Assume that only a single patient $z_j$ out of the ten has $obs(z_j) = 1$. Matching only $x_i$ and $z_j$ and thus leaving the remaining nine possible matches in $B$ unmatched distorts the result. This persists, even if the matching choice was made randomly, as the error lies within the choice of matching only one pair. Note that variants of one-to-many PSM are susceptible to the same error. Bootstrapping avoids this by taking multiple samples, meaning that the matching part of the algorithm is run multiple times. 

As each sample can be perceived as a different permutation of the input, one has to take a high number of samples, which adds an overhead to bootstrapping. Because of this added overhead, the bootstrapping approach seems to be used very rarely. In comparison to the widespread use of PSM, only few studies, e.g., \cite{Knight2016,Chiu2016,Ounpraseuth2012}, make use of bootstrapping with PSM.  We prove that the result of executing PSM with bootstrapping will converge to the result delivered by DBSeM which avoids the overhead of bootstrapping and does not suffer from the remaining pitfalls of PSM.

\section{PSM's Pitfalls}\label{sec:pitfalls}

With regard to the goal of SM it is desirable to establish a matching procedure that delivers identical results for the same input set. We show in this section that results of multiple PSM runs differ significantly even if PSM is applied to the same dataset and identify some of PSM's Pitfalls.

For illustration we use the quality assurance dataset of isolated aortic valve procedures in $2013$, which is an official mandatory dataset including all isolated aortic valve surgery cases in German hospitals and contains patient information (covariates) and mortality information (observed result) for $17{,}427$ patients. For each patient, the corresponding record contains $19$ variables, i.e., $s=19$. This external quality assurance database for isolated aortic valve procedures 2013 of the German Federal Joint Committee contains $9{,}848$ SAVR (replacement surgery of aortic valves) cases and $7{,}579$ TF-AVI cases (transcatheter/transfemoral implantation of aortic valves)\footnote{The cases were documented in accordance with \S 137 Social Security Code V (SGB V) by hospitals registered under \S 108 SGB V. The data collection is compulsory for all in-patient isolated aortic valve procedures in German hospitals.}, held by the Federal Joint Committee (Germany). Given the dataset it can safely be assumed that the data is independent in a statistical sense as patients were only recorded once. The illustrative results, i.e., mortality rates, were calculated using the internationally validated Euroscore~II\footnote{\url{http://www.euroscore.org}} variables and the PSM functions provided by IBM SPSS Statistics for Windows, Version $24.0$.

\subsection{Randomness of Choice and sort order dependence of PSM}\label{subsec:randomness}

For clarification of exposure the following definitions are essential:

\begin{definition}[Sort order]
\label{def:sort_order}
The \emph{sort order} for SM is the order in which patients are ordered in the matrix representing the dataset.
\end{definition}

The following example illustrates the meaning of sort order for SM:

\begin{example}
\label{ex:sort_order}
Let $x_1$ and $x_2$ be patients with covariate vector $cv(x_1) = (1,\,0,\,1)$ and $cv(x_1) = (0,\,1,\,0)$. The order in which $x_1$ and $x_2$ appear in the matrix representing the dataset is the sort order for SM covariates. Thus
\begin{equation*}
\begin{matrix}
 1 & 0 & 1 & (cv(x_1))\\
 0 & 1 & 0 & (cv(x_2)) \\
\end{matrix}
\end{equation*}
and 
\begin{equation*}
\begin{matrix}
0 & 1 & 0 & (cv(x_2))\\
1 & 0 & 1 & (cv(x_1)) \\
\end{matrix}
\end{equation*}
represent different sort orders.
\end{example}

Note that a sort order is valid for the dataset as a whole, thus the whole data matrix is ordered such that a column represents the value of a specific covariate.

Obviously the information contained in a dataset is independent of the sort order of the given dataset. This motivates the following definition:

\begin{definition}
\label{def:sort_order_dep}
An SM-algorithm is sort order dependent if given a dataset with therapy groups $A$ and $B$ the algorithm calculates different results for different sort orders.
\end{definition}

Looking at step~\ref{state:PSM_match} of the general PSM procedure (algorithm~\ref{alg:PSM}), one can infer that if sort order dependence was not in mind and thus taken care of, PSM implementations generally are sort order dependent as the first, or according to a random number, potential match, regardless of the precise matching criteria, i.e., nearest-neighbor, optimal, caliper-matching, gets picked. Additionally the matching of fixed sizes, independent on the exact values of $k$ and $l$, is sort-order dependent as well for the same reason.
%Note that optimal (statistical) matching, as introduced by Rubin \cite{Rosenbaum1989}, does not suffer from this pitfall, but seems to be used very rarely.

The sort order dependency can also be observed by looking at the results from Table~\hyperref[tab:random_runs]{$1$}, which presents PSM calculations on the aforementioned dataset. 

\begin{table}[h!]
	\centering
	\begin{tabular}{l|ll|ll|l}\hline
		$1502$ exact matchings with & \multicolumn{2}{|c|}{SAVR} & \multicolumn{2}{|c|}{TF-AVI} &$\chi^2$ Test \\
		regards to all $19$ Euroscore II& \multicolumn{2}{|c|}{in-hospital death} & \multicolumn{2}{|c|}{ in-hospital death} & (2-tailed)\\
		variables and without replacement& count & \% & count & \% & p-value \\\hline
		Run $1$ & $73$ & $4.9\%$ & $33$ & $2.2\%$ & $<0.0001$\\ 
		Run $2$ & $73$ & $4.9\%$ & $34$ & $2.3\%$ & $<0.0001$\\
		Run $3$ (different sort order) & $42$ & $2.8\%$ & $32$ & $2.1\%$ & $0.2398$\\
		\hline
	\end{tabular}
	\caption{Results of exact 1:1 PSM runs for two heart-surgery methods without bootstrapping}
	\label{tab:random_runs}
\end{table}
\FloatBarrier
The rows labeled Run~$1$ and Run~$3$ (different sort order) differ only in the sort order given in the input. They differ precisely by changing the sort order through ordering one covariate in descending, the patients with $1$ as entry for this covariate come first, instead of ascending order. If PSM would be sort order independent, the result should at least be similar, as the dataset and every other given input was exactly the same. As the results largely differ the possible conclusions drawn from looking at Run~$3$ are contrary to the conclusions one would draw from looking at Run~$1$.

Besides sort order dependence of PSM there is a random element included as well as Run $1$ and Run $2$ used the same sort order, but obtain a slightly different result. The randomness effect occurs for patients $x\in A$ with more than one patient $z \in B$ such that $ps(x) \equiv ps(z)$. For a method used in a scientific context this should not happen as verification of results through reproduction by fellow researchers with the same dataset and software is severely impeded as results are difficult to reproduce.

To clarify the importance of sort order dependence and randomness of choice we calculated the worst and best possible results for exact $1$:$1$ PSM on the given dataset, for results see Table~\ref{tab:best_worst}. The exemplary dataset had mortality as observed values, thus a patient is either dead or alive at the end of the study. Consequently the best case for a therapy group means that living patients from the therapy group were matched to living patients, while avoiding matching living patients to dead patients as long as possible. This can for example be done in the best case for every patient of one partition group, e.g., $A$, by taking the patient's PS and if there is a living patient in $B$ with the same PS, then both living patients get matched. If there is no living patient in $B$, but dead patients with the same PS exist, then they get matched. Naturally patients with different PS do not get matched as we only considered exact matching. One should note that the observed result is not included in the regression model and does not need to be included for simulating a PSM in this manner as patients were only potentially matched if the PS of both patients coincided.

\begin{table}[h!]
\centering
	\begin{tabular}{l|ll|ll|l}\hline
	$1{,}502$ exact matchings with & \multicolumn{2}{|c|}{SAVR} & \multicolumn{2}{|c|}{TF-AVI} &$\chi^2$ Test\footnote{t-test p-values for the first four rows are $<0.0001$ and for PSM with replacement $0.0005$.} \\
	regards to all $19$ Euroscore II& \multicolumn{2}{|c|}{in-hospital death} & \multicolumn{2}{|c|}{ in-hospital death} & (2-tailed)\\
	variables and without replacement& count & \% & count & \% & p-value \\\hline
	Best Case & $24$ & $1.6\%$ & $15$ & $1.0\%$ & $0.1470$\\
	Worst Case & $73$ & $4.9\%$ & $50$ & $3.3\%$ & $0.0342$\\
	Best SAVR/Worst TF-AVI & $24$ & $1.6\%$ & $50$ & $3.3\%$ & $0.0021$\\
	Worst SAVR/Best TF-AVI & $73$ & $4.9\%$ & $15$ & $1.0\%$ & $<0.0001$\\
	Uniform Bootstrapping ($10{,}000$ samples)& $52.47$ & $3.49\%$ & $32.10$ & $2.14\%$ & $0.0210$ (t-test)\\\hline\hline
	PSM with replacement ($3{,}288$ matches) & $73$ & $2.2\%$ & $85$ & $2.5\%$ & $0.3339$ \\
	\end{tabular}
	\caption{Results for exact 1:1 PSM with the same dataset as in Table $1$}
	\label{tab:best_worst}
\end{table}

As the matching procedure was exact PSM the results are balanced regarding the covariates, thus, even if constructed, each of the presented cases is a valid outcome of applying PSM to the dataset. Furthermore the true effect is generally unknown in practice and there is a random element in place. Thus identification of a result as an outlier can be difficult, especially since the balance of these matches is perfect. As the results regarding the observed value is completely different, the conclusions drawn from these results can differ as well. For example most of the medical studies cited in the introduction, e.g., \cite{Burden2017,Ray2012,Nichay2017,Gozalo2015,Mcevoy2016,Capucci2017,Bruno2017,Tranchart2016, Zangbar2016, Dou2017,Fukami2017,McDonald2017,Kishimoto2017,Kong2017,Chen2016,Seung2008,Shaw2008,Liu2016,Svanstrom2013}, given the decision criteria of a $\chi^2$-value above $3.841$, and respectively a p-value below $0.05$, the null hypothesis  ( $H_0$: \emph{The mortality-rate does not depend on therapy}), would be rejected 
for Best SAVR/Worst TF-AVI and Worst SAVR/Best TF-AVI from Table~\ref{tab:best_worst} even though the direction of the results are different, the matchings 
are completely balanced and computed using the same dataset.

Bootstrapping \cite{Austin2014} can solve some of the aforementioned issues if the selection of a matching partner among many is uniform. Thus, we define:
\begin{definition}\label{def:ubPSM}
A bootstrapped PSM is called uniformly bootstrapped PSM (ubPSM) iff the selection choice of patients in $A$ to be matched with a single patient from $B$ of equal PS has the same probability for all patients from $A$ and vice versa.
\end{definition}
Note that if the uniformity assumption made in definition~\ref{def:ubPSM} does not hold, then a bootstrapped result can be skewed, this holds as well if too few bootstrapping iterations were done. Note that this assumption does not hold if one simply applies randomness to the matching procedure. 

Table~\hyperref[tab:best_worst]{$2$} shows the result of applying a ubPSM to our dataset. It is evident that the result significantly differs from some of the other results that were not bootstrapped. In regard of the  pitfall introduced in this section and subsection~\ref{subsec:bootstrap}, it is obvious that PSM with bootstrapping improves result reliability in exchange for computational effort as the change of variance of the result is smaller. An additional drawback of bootstrapping is that one cannot be certain that the drawn amount of samples during the bootstrapping process is large enough. The method shown in section~\ref{sec:DBSeM} of this paper delivers an alternative for this approach and does not suffer from the drawbacks introduced through bootstrapping.

\subsection{Incomplete usage of Data} \label{subsec:incomplete}
For this paragraph suppose that patients $\tilde{x}_1,\,\ldots,\,\tilde{x}_n$ and $\tilde{z}_1,\,\ldots,\,\tilde{z}_m$ with identical PS, $psd(i,\,j) = 0\,\,\forall i \in \{1,\,\ldots,\,n\},\,j\in \{1,\,\ldots,\,m\}$, exist and that $n < m$. 

An exact $1$:$1$ PSM algorithm will create $n$ matching pairs during the matching step, step~\ref{state:PSM_match} in algorithm \ref{alg:PSM}. Therefore $m-n$ many potential matches are ignored and the information provided by the dataset is only incompletely used. As this can, and in practice usually will, happen many times during a single PSM iteration a potentially large amount of information is ignored.

Taking a look at the exemplary calculations the exact $1$:$1$ PSM generates $1502$ matching pairs and thus uses only $15,3\%$ of available SAVR and $19,8\%$ of available TF-AVI-patient data. In section \ref{alg:DBSeM} we will present an algorithm that uses all of the available data and that potentially $34,1\%$ SAVR and $29,7\%$ TF-AVI patients are exact $1$:$1$ matchable.

For the reminder of this paragraph~(\ref{subsec:incomplete}), we will consider $\delta$-matching and assume that the $\tilde{x}_1,\,\ldots,\,\tilde{x}_n$ and $\tilde{z}_1,\,\ldots,\,\tilde{z}_m$ have $psd(i,\,j) \leq \delta\,\forall i \in \{1,\,\ldots,\,n\},\,j\in \{1,\,\ldots,\,m\}$ for given $\delta > 0$ and $n < m$. A $1$:$1$ PSM algorithm will again create at most $n$ matching pairs. Furthermore the larger therapy group usually provides even more potential matching patients for $\delta >0$, thus $n << m$ and the rate of information used is even lower than in the exact matching case.

Note that the shortly discussed $k$:$l$ matching variants, presented in subsection~\ref{subsec:many-one}, will construct at most $n$ matching pairs. Consequently they present no valid solution to this pitfall.

PSM with replacement is supposed to solve the problem of incomplete data usage, but it has the drawback that some patients disproportionally impact the PSM result. This leads to results differing significantly from the outcomes gained through PSM without replacement. This can also be observed by looking at the result presented in the last row of Table~\hyperref[tab:best_worst]{$2$}. While weighting matches according to their frequency~\cite{Stuart2010} alleviates the problem, the distorting nature of PSM with replacement along with the other presented pitfalls persists.

\subsection{Calculation of Propensity Scores}\label{subsec:PSM_calc}

The PS for PSM are typically computed using a type of regression. This results in issues related to floating point comparison, machine precision and the non-uniqueness of solutions of a nonlinear optimization problem. Alongside these issues one has to consider the property stated by proposition \ref{prop:PSM_calc}:

\begin{prop}\label{prop:PSM_calc}
	If no two index sets  $I,\,J \subseteq \{1,\,\ldots,\,s\}$ with $I \neq J$ and the property	
\begin{equation}
\label{eq:PSM_combination}
	\sum_{i \in I} \beta_i = \sum_{j \in J} \beta_j,
\end{equation}
exist, then: Two patients $x,\,z$ have the same covariate vectors, $cv(x) \equiv cv(z)$, if and only if they have the same logistic regression propensity scores, $ps(x) \equiv ps(z)$.
\end{prop}
\begin{proof}
Assume there exist no two index sets $I,\,J$ satisfying equation \eqref{eq:PSM_combination}, but that  $x,\,z$ are two patients with different covariate vectors, $cv(x) \neq cv(z)$, and equal propensity scores, $ps(x) \equiv ps(z)$. Then the following equations lead to a contradiction.
\begin{eqnarray*}
ps(x) = ps(z) &\Leftrightarrow& \frac{e^{\beta_0 + \sum_{j = 1}^{s}\beta_{j}cv_{j}(x)}}{1+e^{\beta_0 + \sum_{j = 1}^{s}\beta_{j}cv_{j}(x)}} = \frac{e^{\beta_0 + \sum_{j = 1}^{s}\beta_{j}cv_{j}(z)}}{1+e^{\beta_0 + \sum_{j = 1}^{s}\beta_{j}cv_{j}(z)}}\\
&\Leftrightarrow& e^{\sum_{j = 1}^{s}\beta_{j}cv_{j}(x)} = e^{\sum_{j = 1}^{s}\beta_{j}cv_{j}(z)}\\
&\Leftrightarrow& \sum_{j = 1}^{s}\beta_{j}cv_{j}(x) = \sum_{j = 1}^{s}\beta_{j}cv_{j}(z)\\
&\Leftrightarrow& cv_j(x) = cv_j(z) \textup{ for all } 1 \leq j\leq s.
\end{eqnarray*}
The last identity holds because by assumption there exists no index sets $I,\,J$ such that equation~\eqref{eq:PSM_combination} holds, thus regression coefficients are unique in the sense of linear combinations. As $cv_j(x),\,cv_j(z) \in \mathbb{R}_{\geq 0}$ this results in a contradiction to the initial assumption that the covariate vectors are different. The opposite direction holds as all relations were equivalent. 
\end{proof}
\bigskip

According to Proposition~\ref{prop:PSM_calc}, PS are not unique if equation~\eqref{eq:PSM_combination} holds for any combination of logistic regression coefficients. Thus, patients with different covariate vectors match despite using exact PSM. This property extends to $\delta$-PSM as one cannot be sure that patients with similar PSs have similar CVs.

This concludes our discussion regarding contribution~\hyperref[C1]{$C1$}. Based on the presented pitfalls, we derive a set of properties which an optimal SM algorithm should have in the next section.

\section{Properties for SM algorithms}\label{sec:properties}

As shown in the previous section, PSM does not compute verifiable and reliable results. Properties~\ref{prop:repro} and~\ref{prop:sort-order} formalize corresponding properties for SM algorithms:

\begin{property}\label{prop:repro}
An SM algorithm has the reproducibility property iff the results given the same input remain exactly the same for any number of computations.
\end{property}

\begin{property}\label{prop:sort-order}
An SM algorithm has the property of sort-order independence iff the result remains the same even if the sort order of covariates of the dataset is changed.
\end{property}

SM algorithms possessing properties~\ref{prop:repro} and~\ref{prop:sort-order} can still produce non-reliable results as they are not necessarily matching in a well defined manner. This is addressed by the following two properties:

\begin{property}\label{prop:complete}	
An exact SM algorithm has the data completeness property, iff for all permutations of patients $\tilde{x}_1,\,\ldots,\,\tilde{x}_n \in A$ and $\tilde{z}_1,\,\ldots,\,\tilde{z}_m \in B$ with identical PS and $m\neq n$, the observed information of all $n+m$ patients has influence on the algorithm's result.
\end{property}

For completeness of exposure we will give an extension of the data completeness property for exact matching to $\delta$-matching here. The extension can be done by introducing a cost function for the matching and the notion of existing possible matches:

\begin{definition}
\label{def:cost_match}
Let $M = \{M_1,\,\ldots,\,M_{|M|}\}$ be a matching and denote the matched patient from therapy group $A$ within the matching set $M_i$ of $M$ with  $M_i(A)$. Then the weight of the matching $M$ is defined by
\begin{equation}
\label{eq:cost_match}
w(M) \coloneqq \sum_{i=1}^{|M|} psd(M_i(A),\,M_i(B)).
\end{equation}
\end{definition}

\begin{definition}
\label{def:poss_match}
A patient $\tilde{x_i} \in A$ is matchable in a $\delta$-matching, if there exists a patient $\tilde{y_j} \in B$ such that $psd(i,\,j) \leq \delta$.
\end{definition}

\begin{definition}
\label{prop:delta_completeness}
A $\delta$-SM algorithm has the data completeness property iff for a matching $M$ and all patients $\tilde{x}_1,\,\ldots,\,\tilde{x}_n \in A$ and $\tilde{z}_1,\,\ldots,\,\tilde{z}_m \in B$ with an existing possible match are matched and $w(M)$ is minimal.
\end{definition}

SM algorithms fulfilling the data completeness property use all information contained in the input as no possible match is ignored. Even PSM with replacement does not have the data completeness property as randomness and sort order dependency still inhibit choosing some possible matches. The last property necessary for an optimal SM matching algorithm guarantees that the determined matching has no additional errors besides the errors stemming from the underlying data.

\begin{property}\label{prop:conserving}
An SM algorithm is called conserving if it is only possible for patients to be matched
\begin{itemize}
	\item in exact matching, if their covariate vectors are the same.
	\item in $\delta$-matching, if their covariates are similar enough according to the chosen similarity measure.
\end{itemize}
\end{property}

While PSM is often assumed to have the conserving property, it is computed using estimated regression scores and this can introduce additional errors as Proposition~\ref{prop:PSM_calc} does not always hold. This concludes our discussion regarding contribution \hyperref[C2]{$C2$} and we present our SM algorithm -- Deterministic Balancing Score exact Matching (DBSeM) -- meeting all four properties for exact SM next.

\section{Deterministic Balancing Score Matching}\label{sec:DBSeM}

The general idea of DBSeM is to cluster patients from a therapy group with same covariate vectors and generate a matching between both therapy groups over the constructed clusters.

Clustering of patients $p$ and $q$ requires a distance metric. In exact matching any metric would be applicable, but for ease of presentation we will use the Manhattan metric $d(p,\,q) := \sum_{i = 1}^{s} \vert cv_i(p) - cv_i(q) \vert$ from now on. Note that patients $p$ and $q$ have equal covariate vectors iff $d(p,\,q) \equiv 0$.

\begin{definition}\label{def:cluster}
A cluster of patients from one therapy group $H$ is a non-empty set $C_{H}$ of patients with properties
\begin{enumerate}
	\item $d(p,\,q) = 0 \,\,\forall p,\,q \in C_{H}$.
	\item $\nexists q \in H$ such that $q\notin C_H$ and $d(p,\,q) = 0$ for $p\in C_H$.
	\item If $p\in C_H$, then the assigned covariate vector of $C_H$ is $cv(p)$.
\end{enumerate}
\end{definition}

Because of definition~\ref{def:cluster} clusters have the following characteristics:

\begin{prop}\label{prop:cluster}
Let $H$ be a therapy group in an SM context, then the following holds for clusters in this therapy group:
\begin{enumerate}
	\item Every patient in $H$ belongs to exactly one cluster.
	\item Every cluster can have exactly one covariate vector assigned to it.
	\item Any two clusters in $H$ have different assigned covariate vectors.
\end{enumerate}
\end{prop}

\begin{proof}
We prove every characteristic individually:
\begin{enumerate}
	\item The assumption that there exists a patient $p\in H$ not belonging to any cluster is by definition \ref{def:cluster} not possible, thus it remains to show that there exists no patient $p\in H$ belonging to two different clusters $C_1$ and $C_2$. Assume that $p\in C_1\cap C_2$ and let $q_1 \in C_1$ and $q_2 \in C_2$ be two patients in $C_1$ and $C_2$ respectively. As $p\in C_1 \cap C_2$ it holds by definition $\hyperref[def:cluster]{\ref{def:cluster}.1}$ that $d(p,\,q_1) = 0 = d(p,\,q_2)$ and therefore $d(q_1,\,q_2)=0$. This is a contradiction to definition $\hyperref[def:cluster]{\ref{def:cluster}.2}$ an therefore every patient belongs to exactly one cluster.
	\item As clusters are non-empty sets of patients every cluster has at least one covariate vector assigned to it. Therefore assume that cluster $C$ has two assigned covariate vectors $v_1$ and $v_2$ differing in at least one entry. Then by definition $\hyperref[def:cluster]{\ref{def:cluster}.3}$ it holds that there exists patients $p,\,q \in C$ such that $v_1 = cv(p)$ and $v_2 = cv(q)$. As $v_1 \neq v_2$ holds by assumption it follows that $d(p,\,q) \neq 0$, contradicting definition $\ref{def:cluster}.1$ as $p,\,q \in C$.
	\item Assume that different clusters $C_1$ and $C_2$ have the same assigned covariate vector. This implies that $d(p,\,q) = 0,\,\forall p\in C_1,\,q\in C_2$ and is a contradiction to definition $\hyperref[def:cluster]{\ref{def:cluster}.2}$.
	\end{enumerate}
\end{proof}
\bigskip

Because of proposition~\ref{prop:cluster}, clusters can be assigned unique covariate vectors. We denote the similarity of two clusters $C_A$ and $C_B$ -- for therapy groups $A$ and $B$ respectively -- as $d(C_A,\,C_B)$. Similarly the distance between a patient $p$ and a cluster $C$ is $d(p,\,C)$.

\begin{prop}\label{prop:cluster_equivalence}
Let $C_A$ and $C_B$ be clusters from different therapy groups, then $d(C_A,\,C_B) \equiv 0$ holds iff the two clusters have the same assigned covariate vector.
\end{prop} 

\begin{proof}
Let $C_A$ and $C_B$ be clusters from different therapy groups and $d(C_A,\,C_B) \equiv 0$. As every cluster has exactly one assigned covariate vector it remains to show that $cv(C_A) \equiv cv(C_B)$ and the following holds:
\begin{equation}
\label{eq:cluster_eq}
d(C_A,\,C_B) \equiv 0 \Leftrightarrow \sum_{i=1}^{s} \vert cv_i(C_A)-cv_i(C_B)\vert \equiv 0 \Leftrightarrow cv_i(C_A) \equiv cv_i(C_B),\,\forall 1\leq i \leq s.
\end{equation}
Thus both clusters have the same assigned covariate vector. The reverse direction follows as all implications in equation~\eqref{eq:cluster_eq} are given through equivalence.
\end{proof}

\paragraph*{The DBSeM algorithm}

Propositions~\ref{prop:cluster} and~\ref{prop:cluster_equivalence} allow us to match clusters in an explicit way and to formulate the following algorithm:

\FloatBarrier
\begin{algorithm}
\caption{DBSeM}
\label{alg:DBSeM}
\begin{algorithmic}[1]
\State Set $c = 0$ and $is\_clustered(x_i) =0$ for all patients in $A$. \label{state:FB_1}
\For{each patient $x_i,\, 1\leq i \leq a$} \label{state:FB_2}
	\If{$is\_clustered(x_i) \equiv 0$} \label{state:FB_3}
		\State Set $c = c+1$, $C_{A,\,c} := \{x_i\}$ and $is\_clustered(x_i) = 1$. \label{state:FB_4}
	\EndIf \label{state:FB_5}
	\For{each patient $x_j$ with $i < j\leq a$ and $is\_clustered(x_j) \equiv 0$} \label{state:FB_6}
		\If{$d(x_j,\,C_{A,\,c})\equiv 0$} \label{state:FB_7}
			\State set $C_{A,\,c} = C_{A,\,c} \cup x_j$ and $is\_clustered(x_j) = 1$ \label{state:FB_8}
		\EndIf \label{state:FB_9}
	\EndFor \label{state:FB_10}
	%\Statex \hspace{12pt} \textbf{for} each patient $x_j$ with $i < j\leq a$ and $is\_clustered(x_j) \equiv 0$ \textbf{do}
	 %\EndIf
\EndFor \label{state:FB_11}
\State Repeat steps $1$ and $2$ for $B$ and store the number of clusters from $A$ and $B$ in variables $k$ and $l$ respectively. \label{state:FB_12}
\For{every cluster $C_{A,\,i},\,1\leq i \leq k$} \label{state:FB_13}
	\State Create Matching Set $M_i = \emptyset$. \label{state:FB_14}
	\State Search for cluster $C_{B,\,c}$ with $d(C_{A,\,i},\,C_{B,\,c}) \equiv 0$. \label{state:FB_15}
	\If{A cluster $C_{B,\,c}$ was found in the previous step} \label{state:FB_16}
	\State Set $M_i = \{C_{A,\,i},\,C_{B,\,c}\}$. \label{state:FB_17}
	\EndIf \label{state:FB_18}
\EndFor \label{state:FB_19}
\State Weight clusters according to a weighting scheme. \label{state:FB_20}
\State Output matching sets $M_k$ and the weighted result. \label{state:FB_21}
\end{algorithmic}
\end{algorithm}
\FloatBarrier

\noindent The weighting in step~\ref{state:FB_20} is required to normalize the results and we will discuss it extensively in the next section. Next, we prove that DBSeM meets the four properties of an optimal SM algorithm.

\begin{theorem}\label{thm:DBSeM}
The \hyperref[alg:DBSeM]{DBSeM algorithm} satisfies properties~\ref{prop:repro} to~\ref{prop:conserving}.
\end{theorem}
\begin{proof} 
We prove reproducibility by contradiction. We assume that two runs of DBSeM generated different matching set results $R_1$ and $R_2$, i.e., different clusters were matched. W.l.o.g. assume that $C \subseteq A$ is matched with $C_1 \subseteq B$ in $R_1$ and $C_2 \subseteq B$ in $R_2$. As $C$ was matched with $C_1$ and $C_2$ we know from Proposition~\ref{prop:cluster_equivalence} that $d(C,\,C_1) \equiv 0 \equiv d(C,\,C_2)$. This implies $d(C_1,\,C_2) \equiv 0$ and $C_1 \equiv C_2$ as of Proposition~\ref{prop:cluster}. Therefore $R_1 \equiv R_2$ as $C,\,C_1$ and $C_2$ were arbitrary. Thus we have a contradiction to the assumption that $R_1$ and $R_2$ were different. The proof for sort-order independence is analogous.

Let $\tilde{x}_1,\,\ldots,\,\tilde{x}_n \in A$ and $\tilde{y}_1,\,\ldots,\,\tilde{y}_n \in B$ be two sets of patients with identical covariate vectors. Because of steps~\ref{state:FB_1} to~\ref{state:FB_12} both patient sets belong to a cluster $C_A$ and $C_B$ respectively. This means $cv(\tilde{x}_i) \equiv cv(\tilde{y}_j),\,\forall 1\leq i \leq n,\,1\leq j\leq m$ and $cv(C_A) \equiv cv(C_B)$. Thus, all patients represented by clusters were matched and impact the matching result. Thus the data completeness property is fulfilled.

The conservation property holds because clusters were only matched if their covariate vectors were the same and every cluster has an unique covariate vector. This concludes the proof as long  as step~\ref{state:FB_20} does not disturb the four properties, which will be proven in proposition~\ref{prop:min_weight}.
\end{proof}
\bigskip

As theorem~\ref{thm:DBSeM} shows DBSeM satisfies the four properties needed for an optimal SM algorithm. According to~\cite{Rosenbaum1983}, the covariate vector is the finest balancing score that expresses differences between patients. Thus, for exact matching one achieves an expression of differences between patients by applying our algorithm. By clustering the patients and comparing matched cluster cardinality, one can estimate assignment biases in both therapies. 
%TODO newly written

Observe that the result given by the DBSeM algorithm is the same as the expected result given by coarsened exact matching (CEM), introduced by \cite{Iacus2011, Iacus2012}, if the strata used in CEM are generated in such a way that a stratum contains all patients with equal covariate vectors from both therapy groups. We stress that the value given by CEM is still an expected value, thus it can change if the algorithm is applied multiple times to the same dataset, while the value given by the DBSeM algorithm is a deterministic one, which is fixed by the data itself and does not change when applying the algorithm multiple times to the same dataset (property \ref{prop:repro}).

Finally note that the result given by algorithm \ref{alg:DBSeM} is imbalance bounded (IB), as defined in \cite{Iacus2012}. It is also equal percent bias reducing (EPBR) \cite{Rubin1976} and we intend to extend our method to $\delta$-matching, with $\delta > 0$, such that these properties (IB) and (EPBR) are kept, while confirming to the four properties introduced in section \ref{sec:properties}.  

This concludes our discussion regarding contribution \hyperref[C3]{$C3$}. Based on the presented algorithm we proceed to present a simple weighting mechanism and prove that bootstrapped ubPSM converges against DBSeM. 

\section{Bootstrapped PSM convergence}\label{sec:convergence}

Step~\ref{state:FB_20} of DBSeM (cf. Algorithm~\ref{alg:min_weight}) uses a weighting approach to avoid that different cardinalities of clusters lead to distorted matching results. In the following we use a min-weighing scheme as it allows us to show convergence of bootstrapped PSM to the DBSeM results. 

The idea is to weight matched clusters $C_{A,\,i}$ and $C_{B,\,j}$ accordingly to their size such that the influence of both clusters is $\min\{\vert C_{A,\,i}\vert,\,\vert C_{B,\,j}\vert\}$ respectively. Algorithm~\ref{alg:min_weight} outlines a min-weighting procedure that needs to be applied to all matched clusters $C_{A,\,i}$ and $C_{B,\,j}$ in step~\ref{state:FB_20} of Algorithm~\ref{alg:DBSeM} (recall that $k$ and $l$ are the number of clusters from $A$ and $B$ respectively). 

\FloatBarrier
\begin{algorithm}
\caption{Min-Weighting Procedure}
\label{alg:min_weight}
\begin{algorithmic}[1]
\State Set $w(C_{A,\,i}) = 0\,\,\forall 1\leq i \leq k$ and $w(C_{B,\,j}) = 0\,\,\forall 1\leq j \leq l$
\For{all $C_{A,\,i},\,1\leq i \leq k$ with $M_i \neq \emptyset$} \label{MW:state_2}
	\Statex Determine the matching cluster $C_{B,\,j},\,1\leq j \leq l$.
	\Statex Calculate $S_{A,\,i} := S_{B,\,j} := \min\{\vert C_{A,\,i}\vert,\,\vert C_{B,\,j}\vert \}$.
	\Statex Compute $w(C_{A,\,i}) := S_{A,\,i}/\vert C_{A,\,i} \vert$ and $w(C_{B,\,j}) := S_{B,\,j}/\vert C_{B,\,j} \vert$.
\EndFor
\State Compute min-weighted results: \label{MW:state_4}
\begin{eqnarray}
		R_A &:=& \sum_{i=1}^{k} [ w(C_{A,\,i}) \sum_{h=1}^{\vert C_{A,\,i}\vert} obs(x_{i,\,h})],\label{eqn:2}\\
		R_B &:=& \sum_{j=1}^{l} [w(C_{B,\,j}) \sum_{h=1}^{\vert C_{B,\,j}\vert} obs(y_{j,\,h})]\label{eqn:3},
	\end{eqnarray} where $x_{i,\,h} \in C_{A,\,i}$ and $y_{j,\,h} \in C_{B,\,j}$.
\end{algorithmic}
\end{algorithm}
\FloatBarrier

\begin{prop}\label{prop:min_weight}
The usage of algorithm~\ref{alg:min_weight} in step~\ref{state:FB_20} of algorithm~\ref{alg:DBSeM} does not disturb the properties of reproducibility, sort-order independence, data completeness and conservation of algorithm~\ref{alg:DBSeM}.
\end{prop}

\begin{proof}
	From the proof of theorem~\ref{thm:DBSeM} we know that steps~\ref{state:FB_1} to~\ref{state:FB_19} of algorithm~\ref{alg:DBSeM} fulfill the properties of reproducibility, sort-order independence, data completeness and conservation. Assume now that algorithm~\ref{alg:min_weight} outputs two different min-weighted results $R_{A,\,1}$ and $R_{A,\,2}$ for therapy group $A$. Then there has to exist at least one pair of matched clusters $C_{A,\,i}$ and $C_{B,\,j}$ with different weights in $R_{A,\,1}$ and $R_{A,\,2}$ as the sum over the observed variables inside a cluster $\sum_{h=1}^{\vert C_{A,\,i}\vert}obs(x_{i,\,h})$ always has the same value and the matched clusters are uniquely matched because of proposition~\ref{prop:cluster_equivalence} and steps~\ref{state:FB_1} to~\ref{state:FB_19} of algorithm~\ref{alg:DBSeM} being reproducible and sort-order independent. As the matched clusters are unique so are their sizes and therefore $S_{A,\,i}$ is unique. Thus $w(C_{A,\,i})$ is the same for both assumed results $R_{A,\,1}$ and $R_{A,\,2}$ and as $C_{A,\,i}$ and $C_{B,\,j}$ were chosen arbitrarily this holds for all clusters. Thus $R_{A,\,1} \equiv R_{A,\,2}$ and the proof is analogous for different results regarding $B$. This proves that the property of reproducibility is not disturbed by using algorithm~\ref{alg:min_weight} in step~\ref{state:FB_20} of algorithm~\ref{alg:DBSeM}. The proof for sort-order independence is analogous. 

	If a patient was inside a matched cluster, then it influences the weight computed in step~\ref{MW:state_2} and the result generated in step~\ref{MW:state_4}. Therefore usage of algorithm~\ref{alg:min_weight} does not disturb algorithm \ref{alg:DBSeM}'s data completeness property.

	As algorithm~\ref{alg:min_weight} does not delete matches, does not match itself and every matched patient is considered, it does not disturb algorithm~\ref{alg:DBSeM}'s conservation property.
\end{proof}
\bigskip

An DBSeM algorithm with the min-weighting procedure in step~\ref{state:FB_20} is called min-weighted DBSeM and as $k\leq \vert A\vert$ and $l\leq \vert B\vert$ the following theorem holds:

\begin{theorem}\label{thm:DBSeM_runtime}
The min-weighted DBSeM algorithm has a runtime of $\mathcal{O}(\vert A \vert \cdot \vert B \vert\cdot s + \vert A\vert^2 + \vert B \vert^2)$.
\end{theorem}

\begin{proof}
In DBSeM step~\ref{state:FB_1} every patient of $A$ gets looked exactly once, while DBSeM steps~\ref{state:FB_2} to~\ref{state:FB_12} have two for-loops and therefore a runtime of $\vert A \vert^2$ and $\vert B \vert^2$ respectively. In DBSeM steps~\ref{state:FB_13} to~\ref{state:FB_19} every cluster in $B$ is investigated at most $\vert A \vert$ times and every comparison between clusters needs $s$ (size of covariate vector) operations to determine the Manhattan metric. This leads to a total runtime of $\mathcal{O}(\vert A \vert \cdot \vert B \vert\cdot s + \vert A\vert^2 + \vert B \vert^2)$ for steps $1$--$4$. Algorithm $3$'s runtime in step~\ref{state:FB_20} is only dependent on the number of clusters $l$ and $k$ in an additive way. As $l\leq \vert A \vert$ and $k\leq \vert B \vert$ it follows that Algorithm $3$ has a runtime of $\mathcal{O}(\max\{\vert A\vert,\,\vert B\vert\})$. Thus the min-weighted DBSeM algorithm has a total runtime of $\mathcal{O}(\vert A \vert \cdot \vert B \vert\cdot s + \vert A\vert^2 + \vert B \vert^2)$.
\end{proof}
\bigskip

Note that the notation given in the statement of theorem \ref{thm:DBSeM_runtime} is due to the fact that we did not assume anything about the sizes of $A,\,B$ or $s$ nor their relative sizes with regard to each other.

Theorem~\ref{thm:DBSeM_convergence} establishes that min-weighted DBSeM has the desirable property of bootstrapped PSM convergence.
As shown in proposition~\ref{prop:PSM_calc}, PSM requires $\beta_i \neq \sum_{j=1,\,j\neq i}^{k} \beta_j$ for all indices $i$ in the logistic regression, to obtain meaningful results, hence we assume this in the following.

\begin{theorem}\label{thm:DBSeM_convergence}
Uniformly bootstrapped $1$:$1$ exact PSM converges towards the outcome of min-weighted-DBSeM.
\end{theorem}

\begin{proof}
	We have to show that the expected values of bootstrapped $1$:$1$ exact PSM results are the same values as in Equations~\eqref{eqn:2} and~\eqref{eqn:3}. Proving convergence towards equality~\eqref{eqn:2} is sufficient, as the proof of~\eqref{eqn:3} follows analogously.
	
	By the law of large numbers it holds that, for a known distribution, the bootstrapped result converges after sufficiently many iterations towards the expected value of the underlying distribution. As expected values for random variables $X$ and $Y$ underlying the same probability distribution are additive, $\mathbf{E}(X+Y) = \mathbf{E}(X) + \mathbf{E}(Y)$, it suffices to identify the distributions and probability for patients in clusters matched by min-weighted DBSeM to be matched by exact PSM.
	
	By assumption the inequality $\beta_i \neq \sum_{j=1,\,j\neq i}^{k} \beta_j\,\forall \beta_i$ holds and we know from Proposition \ref{prop:PSM_calc} that patients with the same propensity score have the same covariate vectors. As we do an exact $1:1$ matching in the PSM part of every bootstrap iteration, the number of patients matched by PSM for a cluster $C_{A,\,i}$ matched with cluster $C_{B,\,j}$ is $S_{A,\,i}$, as their propensity scores are equal. The probability for one patient in $C_{A,\,i}$ to be chosen for matching with a patient from $C_{B,\,j}$ during one bootstrapping iteration is identical for all patients in $C_{A,\,i}$ as we assumed that the selection choice of patients to be matched has the same probability for all patients. Thus we have a discrete uniform distribution over $C_{A,\,i}$ for the matching partner choice in PSM. 	
	
	\noindent It follows that the expected value for cluster $C_{A,\,i}$ matched with $C_{B,\,j}$ calculates as
	\begin{equation}
	\mathbf{E}(C_{A,\,i}) = S_{A,\,i} \cdot (\sum_{h=1}^{\vert C_{A,\,i}\vert}obs(x_{A,\,h}))/\vert C_{A,\,h} \vert.
	\end{equation} 
	Addition of expected values now proves the theorem's statement:
	\begin{eqnarray}
	\mathbf{E}(A) &=& \sum_{i = 1}^{k}\mathbf{E}(C_{A,\,i}) = \sum_{i = 1}^{k} S_{A,\,i} \cdot (\sum_{h=1}^{\vert C_{A,\,i}\vert}obs(x_{A,\,h}))/\vert C_{A,\,i} \vert\\
	&=& \sum_{i = 1}^{k}= S_{A,\,i}/\vert C_{A,\,i}\vert \sum_{h=1}^{\vert C_{A,\,i}\vert}obs(x_{A,\,h}) = \sum_{i = 1}^{k} w(C_{A,\,i})\sum_{h=1}^{\vert C_{A,\,i}\vert}obs(x_{A,\,h})\\
	&=& R_A.
	\end{eqnarray}
\end{proof}
\bigskip

Table~\hyperref[tab:DBSeM]{$3$} shows the result for min-weighted DBSeM with our dataset from Tables~\hyperref[tab:random_runs]{$1$} and~\hyperref[tab:best_worst]{$2$}. The DBSeM result is close but not equal to the result obtained uniformly bootstrapped PSM in Table~\hyperref[tab:best_worst]{$2$}. This is because even with bootstrapping
\begin{enumerate}
	\item some information is lost during the matching (not all possible matches are used) and
	\item some matchings are overrepresented, i.e., sampled more than once.
\end{enumerate} 
 Uniformly bootstrapped PSM will only achieve the exact same result as DBSeM if all permutations of the possible different matching samples are used exactly the same number of times (cf. Theorem~\ref{thm:DBSeM_convergence}). Since DBSeM has the data completeness property and PSM does not, the result in Table~\hyperref[tab:DBSeM]{$3$} represents the ground truth that PSM can only achieve with bootstrapping through all matching permutations. In general, there are $(\max\{a,\,b\})!$ such permutations which makes computing PSM for all of them not feasible. Hence, DBSeM performs better in SM compared to exact PSM as PSM would need a very large amount of iterations to generate the same result with a bootstrapping approach.

\begin{table}[h!]
\centering
	\begin{tabular}{l|ll|ll|l}\hline
		$1{,}502$ matched clusters with & \multicolumn{2}{|c|}{SAVR} & \multicolumn{2}{|c|}{TF-AVI} &t-test \\
		regards to all $19$ Euroscore II& \multicolumn{2}{|c|}{in-hospital death} & \multicolumn{2}{|c|}{ in-hospital death} & (2-tailed)\\
		variables and without replacement & count & \% & count & \% & p-value \\\hline
		Min-weighted DBSeM & $53.01$ & $ 3.5\%$& $32.32$ & $2.1\%$ &  $0.02271$\\
	\end{tabular}
	\caption{Results for min-weighted DBSeM with the same dataset as in Tables $1$ and $2$}
	\label{tab:DBSeM}
\end{table}

We conclude with some remarks for practitioners and comment on the scope of our contribution.

We have shown that PSM delivers non-reliable and non-reproducible results~\hyperref[C1]{($C1$)} and formally deduced four properties for optimal SM algorithms~\hyperref[C1]{($C2$)}. The proposed DBSeM procedure meets the four derived formal properties for optimal SM algorithms~\hyperref[C1]{($C3$)} and delivers as the result the average of all valid sets of matched pairs for the investigated dataset, while being computationally very efficient~\hyperref[C1]{($C4$)}. 

The presented DBSeM-algorithm can be used to support results, generated through other methods, e.g. PSM, CEM. As the result given by DBSeM is deterministic for a given dataset, and therefore definite, see Theorem \ref{thm:DBSeM}, it is possible to use the result for verification as the exact matching should be part of every $\delta$-matching with $\delta>0$. If the observational results of the DBSeM-matching and the chosen $\delta$-matching method coincide, then the quality of the calculated $\delta$-matching is more likely to be good in the sense of statistical matching criteria such as (EPBR) and (IB). On the other hand if the results contradict each other the practitioner should consider the collection of additional data.

Further work in regards to the presented method is the extension of DBSeM, such that $\delta$-matchings for $\delta>0$ can be constructed through a deterministic method as well.


\begin{thebibliography}{10}

\bibitem{Rubin1973}
Rubin DB.
\newblock Matching to Remove Bias in Observational Studies.
\newblock Biometrics. 1973;29(1):159--183.

\bibitem{Anderson1980}
Anderson DW, Kish L, Cornell RG.
\newblock On Stratification, Grouping and Matching.
\newblock Scandinavian Journal of Statistics. 1980;7(2):61--66.

\bibitem{Kupper1981}
Kupper LL, Karon JM, Kleinbaum DG, Morgenstern H, Lewis DK.
\newblock Matching in Epidemiologic Studies: Validity and Efficiency
  Considerations.
\newblock Biometrics. 1981;37(2):271--291.

\bibitem{Ray2012}
Ray WA, Murray KT, Hall K, Arbogast PG, Stein CM.
\newblock Azithromycin and the Risk of Cardiovascular Death.
\newblock New England Journal of Medicine. 2012;366(20):1881--1890.
\newblock doi:{10.1056/NEJMoa1003833}.

\bibitem{Zhang2015}
Zhang Z, Chen K, Ni H.
\newblock Calcium supplementation improves clinical outcome in intensive care
  unit patients: a propensity score matched analysis of a large clinical
  database MIMIC-II.
\newblock SpringerPlus. 2015;4:594.
\newblock doi:{10.1186/s40064-015-1387-7}.

\bibitem{Gozalo2015}
Gozalo P, Plotzke M, Mor V, Miller SC, Teno JM.
\newblock Changes in Medicare Costs with the Growth of Hospice Care in Nursing
  Homes.
\newblock New England Journal of Medicine. 2015;372(19):1823--1831.
\newblock doi:{10.1056/NEJMsa1408705}.

\bibitem{Zhang2016}
Zhang M, Guddeti RR, Matsuzawa Y, Sara JDS, Kwon TG, Liu Z, et~al.
\newblock Left Internal Mammary Artery Versus Coronary Stents: Impact on
  Downstream Coronary Stenoses and Conduit Patency.
\newblock Journal of the American Heart Association. 2016;5(9).
\newblock doi:{10.1161/JAHA.116.003568}.

\bibitem{Cho2016}
Cho SH, Choi GS, Kim GC, Seo AN, Kim HJ, Kim WH, et~al.
\newblock Long-term outcomes of surgery alone versus surgery following
  preoperative chemoradiotherapy for early T3 rectal cancer: A propensity score
  analysis.
\newblock Medicine. 2017;96(12):e6362.
\newblock doi:{10.1097/md.0000000000006362}.

\bibitem{Bruno2017}
Bruno S, Marco VD, Iavarone M, Roffi L, Boccaccio V, Crosignani A, et~al.
\newblock Improved survival of patients with hepatocellular carcinoma and
  compensated hepatitis C virus-related cirrhosis who attained sustained
  virological response.
\newblock Liver International. 2017;37(10):1526--1534.
\newblock doi:{10.1111/liv.13452}.

\bibitem{Nichay2017}
Nichay NR, Gorbatykh YN, Kornilov IA, Soynov IA, Ivantsov SM, Gorbatykh AV,
  et~al.
\newblock Bidirectional cavopulmonary anastomosis with additional pulmonary
  blood flow: good or bad pre-Fontan strategy?
\newblock Interactive CardioVascular and Thoracic Surgery. 2017;24(4):582--589.
\newblock doi:{10.1093/icvts/ivw429}.

\bibitem{Burden2017}
Burden A, Roche N, Miglio C, Hillyer E, Postma D, Herings R, et~al.
\newblock An evaluation of exact matching and propensity score methods as
  applied in a comparative effectiveness study of inhaled corticosteroids in
  asthma.
\newblock Pragmatic and Observational Research. 2017;8:15-30.
\newblock doi:{10.2147/POR.S122563}.

\bibitem{Mcevoy2016}
McEvoy RD, Antic NA, Heeley E, Luo Y, Ou Q, Zhang X, et~al.
\newblock CPAP for Prevention of Cardiovascular Events in Obstructive Sleep
  Apnea.
\newblock New England Journal of Medicine. 2016;375(10):919--931.
\newblock doi:{10.1056/NEJMoa1606599}.

\bibitem{Schermerhorn2008}
Schermerhorn ML, O'Malley AJ, Jhaveri A, Cotterill P, Pomposelli F, Landon BE.
\newblock Endovascular vs. Open Repair of Abdominal Aortic Aneurysms in the
  Medicare Population.
\newblock New England Journal of Medicine. 2008;358(5):464--474.
\newblock doi:{10.1056/NEJMoa0707348}.

\bibitem{Lee2017}
Lee SI, Lee KS, Kim JB, Choo SJ, Chung CH, Lee JW, et~al.
\newblock Early Antithrombotic Therapy after Bioprosthetic Aortic Valve
  Replacement in Elderly Patients: A Single-Center Experience.
\newblock Annals of Thoracic and Cardiovascular Surgery. 2017;23(3):128--134.
\newblock doi:{10.5761/atcs.oa.16-00297}.

\bibitem{Capucci2017}
Capucci A, De~Simone A, Luzi M, Calvi V, Stabile G, D'Onofrio A, et~al.
\newblock Economic impact of remote monitoring after implantable defibrillators
  implantation in heart failure patients: an analysis from the EFFECT study.
\newblock EP Europace. 2017;19(9):1493--1499.
\newblock doi:{10.1093/europace/eux017}.

\bibitem{Tranchart2016}
Tranchart H, Fuks D, Vigano L, Ferretti S, Paye F, Wakabayashi G, et~al.
\newblock Laparoscopic simultaneous resection of colorectal primary tumor and
  liver metastases: a propensity score matching analysis.
\newblock Surgical Endoscopy. 2016;30(5):1853--1862.
\newblock doi:{10.1007/s00464-015-4467-4}.

\bibitem{Zangbar2016}
Zangbar B, Khalil M, Gruessner A, Joseph B, Friese R, Kulvatunyou N, et~al.
\newblock Levetiracetam Prophylaxis for Post-traumatic Brain Injury Seizures is
  Ineffective: A Propensity Score Analysis.
\newblock World Journal of Surgery. 2016;40(11):2667--2672.
\newblock doi:{10.1007/s00268-016-3606-y}.

\bibitem{Dou2017}
Dou JP, Yu J, Yang XH, Cheng ZG, Han ZY, Liu FY, et~al.
\newblock Outcomes of microwave ablation for hepatocellular carcinoma adjacent
  to large vessels: a propensity score analysis.
\newblock Oncotarget. 2017;8(17):28758--28768.
\newblock doi:{10.18632/oncotarget.15672}.

\bibitem{Fukami2017}
Fukami H, Takeuchi Y, Kagaya S, Ojima Y, Saito A, Sato H, et~al.
\newblock Perirenal fat stranding is not a powerful diagnostic tool for acute
  pyelonephritis.
\newblock International Journal of General Medicine. 2017;Volume 10:137--144.
\newblock doi:{10.2147/ijgm.s133685}.

\bibitem{McDonald2017}
McDonald JS, McDonald RJ, Williamson EE, Kallmes DF, Kashani K.
\newblock Post-contrast acute kidney injury in intensive care unit patients: a
  propensity score-adjusted study.
\newblock Intensive Care Medicine. 2017;43(6):774--784.
\newblock doi:{10.1007/s00134-017-4699-y}.

\bibitem{Lai2016}
Lai WH, Rau CS, Wu SC, Chen YC, Kuo PJ, Hsu SY, et~al.
\newblock Post-traumatic acute kidney injury: a cross-sectional study of trauma
  patients.
\newblock Scandinavian Journal of Trauma, Resuscitation and Emergency Medicine.
  2016;24(1):136.
\newblock doi:{10.1186/s13049-016-0330-4}.

\bibitem{Abidov2005}
Abidov A, Rozanski A, Hachamovitch R, Hayes SW, Aboul-Enein F, Cohen I, et~al.
\newblock Prognostic Significance of Dyspnea in Patients Referred for Cardiac
  Stress Testing.
\newblock New England Journal of Medicine. 2005;353(18):1889--1898.
\newblock doi:{10.1056/NEJMoa042741}.

\bibitem{Adams2017}
Adams N, Gibbons KS, Tudehope D.
\newblock Public-private differences in short-term neonatal outcomes following
  birth by prelabour caesarean section at early and full term.
\newblock Australian and New Zealand Journal of Obstetrics and Gynaecology.
  2017;57(2):176--185.
\newblock doi:{10.1111/ajo.12591}.

\bibitem{Kishimoto2017}
Kishimoto M, Yamana H, Inoue S, Noda T, Myojin T, Matsui H, et~al.
\newblock Sivelestat sodium and mortality in pneumonia patients requiring
  mechanical ventilation: propensity score analysis of a Japanese nationwide
  database.
\newblock Journal of Anesthesia. 2017;31(3):405--412.
\newblock doi:{10.1007/s00540-017-2327-1}.

\bibitem{Kong2017}
Kong L, Li M, Li L, Jiang L, Yang J, Yan L.
\newblock Splenectomy before adult liver transplantation: a retrospective
  study.
\newblock BMC Surgery. 2017;17(1):44.
\newblock doi:{10.1186/s12893-017-0243-9}.

\bibitem{Chen2016}
Chen HY, Wang Q, Xu QH, Yan L, Gao XF, Lu YH, et~al.
\newblock Statin as a Combined Therapy for Advanced-Stage Ovarian Cancer: A
  Propensity Score Matched Analysis.
\newblock {BioMed} Research International. 2016;2016:1--5.
\newblock doi:{10.1155/2016/9125238}.

\bibitem{Seung2008}
Seung KB, Park DW, Kim YH, Lee SW, Lee CW, Hong MK, et~al.
\newblock Stents versus Coronary-Artery Bypass Grafting for Left Main Coronary
  Artery Disease.
\newblock New England Journal of Medicine. 2008;358(17):1781--1792.
\newblock doi:{10.1056/NEJMoa0801441}.

\bibitem{Shaw2008}
Shaw AD, Stafford-Smith M, White WD, Phillips-Bute B, Swaminathan M, Milano C,
  et~al.
\newblock The Effect of Aprotinin on Outcome after Coronary-Artery Bypass
  Grafting.
\newblock New England Journal of Medicine. 2008;358(8):784--793.
\newblock doi:{10.1056/NEJMoa0707768}.

\bibitem{Liu2016}
Liu Y, Han J, Liu T, Yang Z, Jiang H, Wang H.
\newblock The Effects of Diabetes Mellitus in Patients Undergoing Off-Pump
  Coronary Artery Bypass Grafting.
\newblock {BioMed} Research International. 2016;2016:1--6.
\newblock doi:{10.1155/2016/4967275}.

\bibitem{Svanstrom2013}
Svanstr\"om H, Pasternak B, Hviid A.
\newblock Use of Azithromycin and Death from Cardiovascular Causes.
\newblock New England Journal of Medicine. 2013;368(18):1704--1712.
\newblock doi:{10.1056/NEJMoa1300799}.

\bibitem{Salati2017}
Salati M, Brunelli A, Xium\`e F, Monteverde M, Sabbatini A, Tiberi M, et~al.
\newblock Video-assisted thoracic surgery lobectomy does not offer any
  functional recovery advantage in comparison to the open approach 3 months
  after the operation: a case matched analysis?
\newblock European Journal of Cardio-Thoracic Surgery. 2017;51(6):1177--1182.
\newblock doi:{10.1093/ejcts/ezx013}.

\bibitem{King2016}
King G, Nielsen R. Why propensity scores should not be used for matching. 2016;
  2015.

\bibitem{Rosenbaum1983}
Rosenbaum PR~and Rubin DB.
\newblock The central role of the propensity score in observational studies for
  causal effects.
\newblock Biometrika. 1983;70(1):41--55.
\newblock doi:{10.1093/biomet/70.1.41}.

\bibitem{Austin2011}
Austin PC.
\newblock An Introduction to Propensity Score Methods for Reducing the Effects
  of Confounding in Observational Studies.
\newblock Multivariate Behavioral Research. 2011;46(3):399--424.
\newblock doi:{10.1080/00273171.2011.568786}.

\bibitem{Pearl2009}
Pearl J.
\newblock Causality: Models, Reasoning and Inference.
\newblock 2nd ed. New York, NY, USA: Cambridge University Press; 2009.

\bibitem{Stuart2014}
Stuart EA, Huskamp HA, Duckworth K, Simmons J, Song Z, Chernew ME, et~al.
\newblock Using propensity scores in difference-in-differences models to
  estimate the effects of a policy change.
\newblock Health Services and Outcomes Research Methodology.
  2014;14(4):166--182.
\newblock doi:{10.1007/s10742-014-0123-z}.

\bibitem{Stuart2010}
Stuart EA.
\newblock Matching Methods for Causal Inference: A Review and a Look Forward.
\newblock Statistical Science. 2010;25(1):1--21.
\newblock doi:{10.1214/09-sts313}.

\bibitem{Iacus2011}
Iacus SM, King G, Porro G.
\newblock Causal inference without balance checking: Coarsened exact matching.
\newblock Political analysis. 2012;20(1):1--24.

\bibitem{Caliendo2005}
Caliendo M, Kopeinig S.
\newblock Some pracitcal Guidance for the implementation of propensity score
  matching.
\newblock Journal of Economic Surveys. 2008;22(1):31--72.
\newblock doi:{10.1111/j.1467-6419.2007.00527.x}.

\bibitem{Rubin1974}
Rubin DB.
\newblock Estimating causal effects of treatments in randomized and
  nonrandomized studies.
\newblock Journal of Educational Psychology. 1974;66.

\bibitem{Rosenbaum1989}
Rosenbaum PR. 
\newblock Optimal Matching for Observational Studies.
\newblock Journal of the American Statistical Association. 1989;84(408):1024--1032.
\newblock doi:{doi:10.2307/2290079}.

\bibitem{Austin2014}
Austin PC, Small DS.
\newblock The use of bootstrapping when using propensity-score matching without
  replacement: a simulation study.
\newblock Statistics in Medicine. 2014;33(24):4306--4319.
\newblock doi:{10.1002/sim.6276}.

\bibitem{Knight2016}
Knight SR, Oniscu GC, Devey L, Simpson KJ, Wigmore SJ, Harrison EM.
\newblock Use of Renal Replacement Therapy May Influence Graft Outcomes
  following Liver Transplantation for Acute Liver Failure: A Propensity-Score
  Matched Population-Based Retrospective Cohort Study.
\newblock PLOS ONE. 2016;11(3):1--14.
\newblock doi:{10.1371/journal.pone.0148782}.

\bibitem{Chiu2016}
Chiu M, Rezai MR, Maclagan LC, Austin PC, Shah BR, Redelmeier DA, et~al.
\newblock Abstract 11545: Moving to a Highly Walkable Neighborhood and
  Incidence of Hypertension: A Propensity-score Matched Cohort Study.
\newblock Circulation. 2015;132(Suppl 3):A11545--A11545.

\bibitem{Ounpraseuth2012}
Ounpraseuth S, Gauss CH, Bronstein J, Lowery C, Nugent R, Hall R.
\newblock Evaluating the Effect of Hospital and Insurance Type on the Risk of
  1-year Mortality of Very Low Birth Weight Infants.
\newblock Medical Care. 2012;50(4):353--360.
\newblock doi:{10.1097/mlr.0b013e318245a128}.

\bibitem{Iacus2012}
Iacus SM, King G, Porro G, N. Katz J.
\newblock Causal Inference Without Balance Checking: Coarsened Exact Matching
\newblock Political Analysis. 2012;20:1--24.

\bibitem{Rubin1976}
Rubin DB.
\newblock Multivariate Matching Methods That are Equal Percent Bias Reducing, I: Some Examples.
\newblock Biometrics. 1976; 1(32):109--120.


\end{thebibliography}
\end{document}